\documentclass[english]{article}
\usepackage[T1]{fontenc}
\usepackage[latin9]{inputenc}
\usepackage{amsthm}
\usepackage{amsmath}
\usepackage{amssymb}
\usepackage{graphicx}
\PassOptionsToPackage{normalem}{ulem}
\usepackage{ulem}

\makeatletter
\numberwithin{equation}{section}
\numberwithin{figure}{section}
\theoremstyle{plain}
\newtheorem{thm}{\protect\theoremname}
  \theoremstyle{definition}
  \newtheorem{defn}[thm]{\protect\definitionname}
  \theoremstyle{plain}
  \newtheorem{lem}[thm]{\protect\lemmaname}
  \theoremstyle{remark}
  \newtheorem*{rem*}{\protect\remarkname}
  \theoremstyle{plain}
  \newtheorem*{thm*}{\protect\theoremname}
  \theoremstyle{plain}
  \newtheorem*{conjecture*}{\protect\conjecturename}
  \theoremstyle{definition}
  \newtheorem*{problem*}{\protect\problemname}

\makeatother

\usepackage{babel}
  \providecommand{\conjecturename}{Conjecture}
  \providecommand{\definitionname}{Definition}
  \providecommand{\lemmaname}{Lemma}
  \providecommand{\problemname}{Problem}
  \providecommand{\remarkname}{Remark}
  \providecommand{\theoremname}{Theorem}
\providecommand{\theoremname}{Theorem}

\begin{document}

\title{Braid\emph{ }is undecidable}

\author{Linus Hamilton}
\maketitle
\begin{abstract}
Braid is a 2008 puzzle game centered around the ability to reverse
time. We show that Braid can simulate an arbitrary computation. Our
construction makes no use of Braid's unique time mechanics, and therefore
may apply to many other video games.

We also show that a plausible ``bounded'' variant of Braid lies
within $\mathsf{2}\text{-}\mathsf{EXPSPACE}$. Our proof relies on
a technical lemma about Turing machines which may be of independent
interest. Namely, define a \emph{braidlike }Turing machine to be a
Turing machine that, when it writes to the tape, deletes all data
on the tape to the right of the head. We prove that deciding the behavior
of such a machine lies in $\mathsf{EXPSPACE}$.
\end{abstract}

\section{Introduction}

Many video games have recently been proven $\mathsf{NP}$-complete,
$\mathsf{NP}$-hard, $\mathsf{PSPACE}$-complete, and more \cite{DBLP:journals/corr/abs-1203-1895,DBLP:journals/corr/abs-1201-4995,DBLP:journals/corr/abs-1202-6581}.
However, to this author's knowledge, no one has ever proven a video
game to be outside of $\mathsf{PSPACE}$, much less formally undecidable.
This is likely due to Savitch's Theorem, which asserts that any video
game that can be simulated in polynomial space can also be solved
in polynomial space. Most commercial video games yield to Savitch's
Theorem. However, Braid does not, for reasons we explain in Section
3.

Before delving into the computational complexity of Braid, we will
explain the important game elements, along with the unique time-rewinding
mechanic.

\section{A Guide to Braid}

\subsection{Important Elements}

\begin{center}
\includegraphics[scale=0.5]{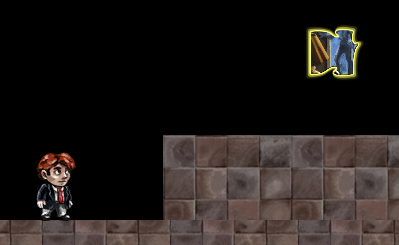}
\par\end{center}

This is Tim. The ledge in front of him shows his maximum jump height.
On top of the ledge is a puzzle piece. The goal of the game is to
collect it.

\begin{center}
\includegraphics[scale=0.5]{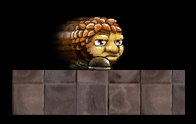}
\par\end{center}

This is a ``monstar''. It behaves like a Goomba from Super Mario
Bros. It walks forward, falls off ledges, and turns around when it
bumps into something. Tim can bounce on it to jump higher.

\begin{center}
\includegraphics[scale=0.5]{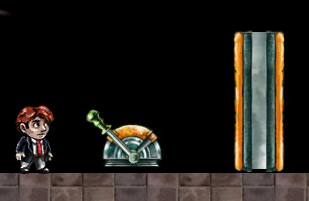}
\par\end{center}

This is a lever and a platform. Tim can pull levers in order to control
platforms in various ways. They are useful for gadgets where we need
a door to open and close.

It is possible to set a platform so that it rises when Tim pulls a
lever, and falls back down after it hits something. We use this behavior
many times to build multi-use gadgets.

\begin{center}
\includegraphics[scale=0.5]{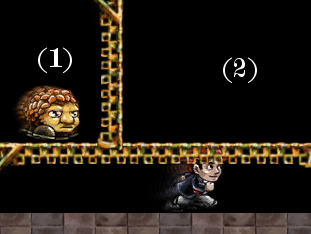}
\par\end{center}

These are one-way surfaces. Tim can jump to (2) onto the surface,
but cannot fall back down. The monstar will walk from (1) to (2),
but cannot return.

\begin{center}
\includegraphics[scale=0.5]{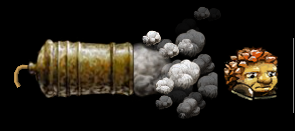} 
\par\end{center}

This is a cannon. It continually shoots out one of monstars, bunnies,
fireballs, or clouds. In this paper, we use only monstar cannons and
bunny cannons.

\begin{center}
\includegraphics[scale=0.5]{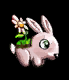}
\par\end{center}

This is a bunny. The only thing you need to know about it is that
monstars can bounce off of bunnies, killing them. We use this behavior
in order to separate one monstar from a large crowd of monstars.

\subsection{Rewinding}

The player can hold the Shift key to rewind time, undoing his mistakes
and even reverting his own death. Ordinarily, this mechanic would
make a game more merciful for human players, but not alter its computational
complexity. However, in Braid, some objects live outside of time.
These objects are unaffected by time manipulation. A time-immune bunny,
for example, will keep hopping forward, even when the entire world
around it is moving backwards through time. This mechanic forms the
basis for all of Braid's unique puzzles.

After rewinding time, Tim can ``play it forward'' again. This works
analogously to the ``undo'' and ``redo'' buttons in a computer
program. Here is an example:

\begin{center}
\includegraphics[scale=0.5]{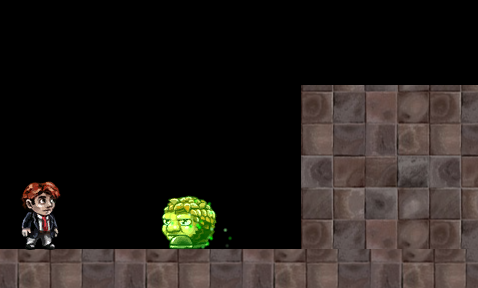}
\par\end{center}

Here, the monstar is sparkling green, indicating that it lives outside
of time.

Tim can jump on the monstar, killing it and bouncing up to the upper
ledge. After that, Tim can rewind time, rewinding himself back down
to the lower ledge. In the computer program analogy, this is like
pressing ``undo''. The monstar does not reappear: it lives outside
of time, and is therefore permanently dead.

At this point, Tim can ``play time forward''. Tim redoes his jump,
even though there is no longer anything to jump off of, and returns
to the upper ledge. This is like pressing ``redo''.

However, if the player rewinds time, and then releases the Shift key
to start controlling Tim again, the ``redo'' future is deleted forever.
Analogously, in a computer program, if you ``undo'' something and
then take another action, then the ``redo'' option is lost forever.

\section{Braid's Computational Complexity: A Gentle Introduction}

Before delving into our main results, we will demonstrate that Braid
is $\mathsf{PSPACE}$-hard.

\begin{center}
\includegraphics[scale=0.33]{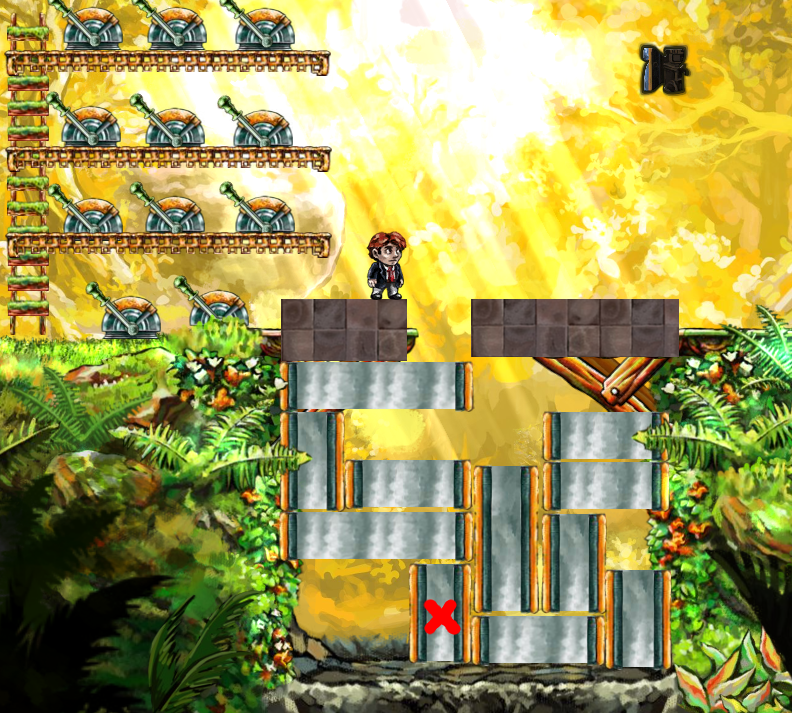}
\par\end{center}

Braid contains levers which can control platforms. Using this, we
can simulate a variant of Rush Hour, as in the picture above. Tim
must beat the Rush Hour puzzle, releasing the marked platform, in
order to jump on it and reach the puzzle piece.

This variant of Rush Hour is $\mathsf{PSPACE}$-hard \cite{Flake_rushhour},
so Braid is too.

For most video games, the article would end here. Savitch's Theorem
states that, if a video game can be simulated using polynomial memory,
then it can be solved in $\mathsf{PSPACE}$. This logic places almost
every video game in $\mathsf{PSPACE}$.

However, this logic does not apply to Braid. In order to simulate
Braid, one might need an arbitrarily large amount of memory. The game
has cannons which can spawn an arbitrary number of enemies. In addition,
Tim can rewind to any previous point in time, so the game must keep
track of the entire previous timeline.

This opens up a rather intriguing possibility. Concievably, one might
use Braid's rewind data as the tape of a Turing machine. By rewinding
time and playing it back, Tim could act as the Turing machine's head.
Perhaps one can use this to simulate an arbitrary computation.

As it turns out, Braid \emph{can} simulate an arbitrary computation
-- but not because of time manipulation.

\section{Braid is undecidable}

In this section, we will prove that Braid is formally undecidable.
To do this, we will demonstrate how to use Braid to simulate a type
of machine called a \emph{counter machine}. Counter machines are one
of the simplest types of machines known to exhibit Turing-complete
behavior.

\subsection{What are Counter Machines?}

A counter machine consists of several counters, each of which holds
a non-negative integer. The machine also contains a program to manipulate
these counters. For the purposes of this paper, a program has three
types of operation:

$ $

1. ``Add''. Adds one to the value of the specified counter.

2. ``Subtract and Branch''. Subtracts one from the value of the
specified counter, if it was not 0. If the counter's value was 0,
goto a specified line of the program instead.

3. ``Halt''.

$ $

A counter machine is extremely simple. Surprisingly, however, a counter
machine with only three counters is Turing-complete \cite{counters}.

\subsection{Overview of the proof}

Our goal is to simulate a counter machine in Braid. To do this, each
counter will be simulated by a stack of overlapping monstars. The
machine's program will be executed by Tim himself. To beat the level,
Tim must run along a course and pull all the levers he encounters.
Each lever performs either an ``Add'' operation or a ``Subtract
and Branch'' operation to a specific counter. In order to allow gotos,
parts of the course may loop back to earlier parts of the course.

In Section 4.3, we design the gadgets that make all this possible.

\subsection{Gadgets}

\subsubsection*{Lever Pull Gadget}

\begin{center}
\includegraphics[scale=0.5]{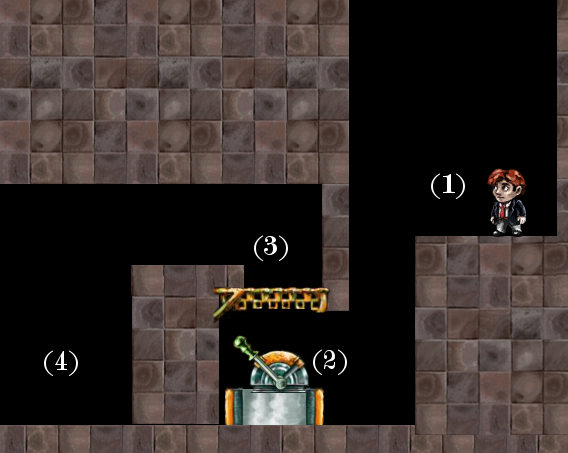}
\par\end{center}

This gadget forces Tim to pull a lever exactly once to progress.

Tim starts at (1). His only option is to fall to (2) and pull the
lever.

When Tim pulls the lever, the platform below it moves upward, lifting
Tim to (3). It moves at such a fast rate that Tim has no time to pull
the lever again.

There is a one-way surface at (3) which can only be passed in the
upward direction. Therefore, once at (3), Tim's only option is to
proceed to (4) and exit the gadget.

The platform is set so that, after it reaches (3), it falls back down
to its original position. Therefore, this gadget may be traversed
as many times as necessary.

$ $

In Braid, levers can control more than one platform, and platforms
can be controlled by more than one lever. We will use this gadget
to force Tim to control platform machinery in distant parts of the
level.

\subsubsection*{Crossover Gadget}

As with any computational complexity analysis of a 2D video game,
we need a Crossover Gadget. This has no purpose in the actual construction;
it is merely to ensure that we can force Tim and monstars down whatever
complicated paths we want to, despite being confined to a 2D plane.

In Braid, a Crossover Gadget is simple to construct, both for Tim
and for monstars.

\begin{center}
\includegraphics[scale=0.5]{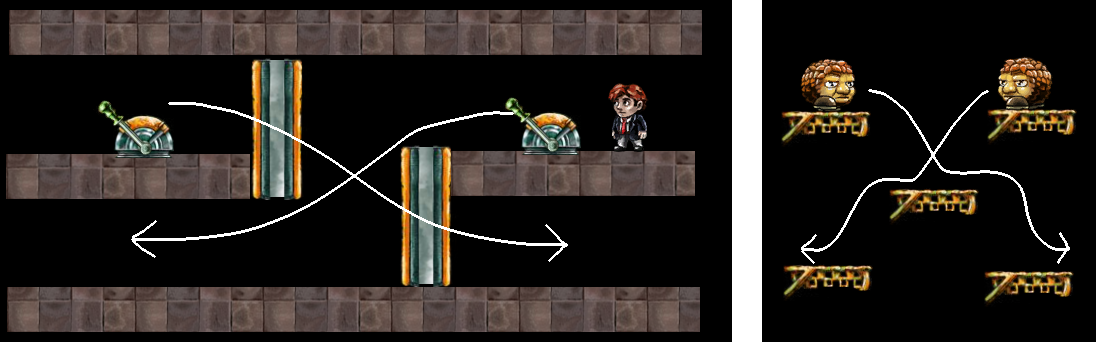}
\par\end{center}

Both levers in Tim's crossover gadget toggle both platforms between
up and down.

Finally, here is a Crossover Gadget to cross Tim's path with a monstar's
path. It works because Tim is too tall to fit in tight spaces.

\begin{center}
\includegraphics[scale=0.33]{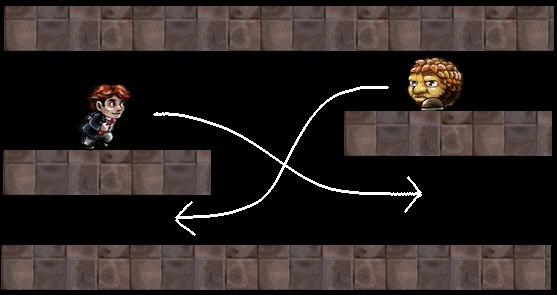}
\par\end{center}

When using this last gadget, one must take care that Tim cannot arrive
at the gadget at the same time as a monstar and interfere with its
path. For example, one can make Tim's path to the gadget very long,
so that if a monstar is coming, it will always arrive first.

\subsubsection*{Counter Gadget}

The counter is a bit complicated, and composed of multiple sub-gadgets.
We will describe the machinery one part at a time.

\begin{center}
\includegraphics[scale=0.5]{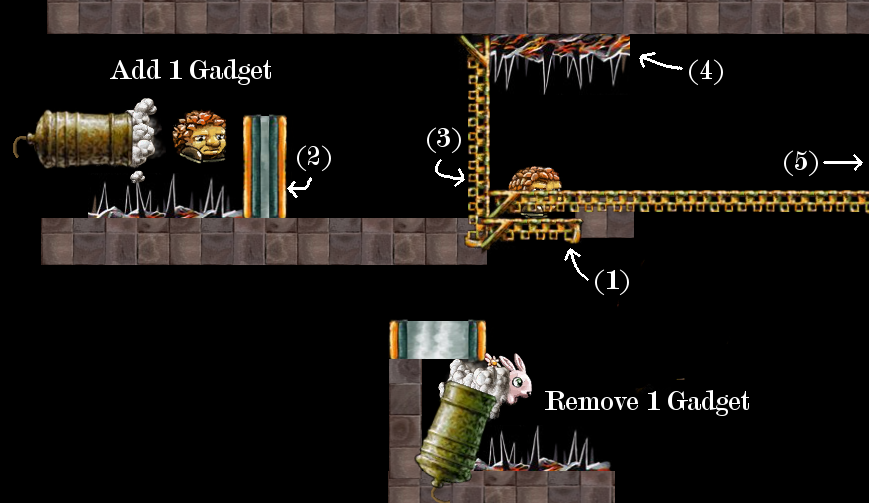}
\par\end{center}

\emph{Normal operation: }Many monsters are trapped walking back and
forth on surface (1). When $n$ monstars are trapped there, this means
the current value of the counter is $n$.

If monstars could jump, then they could jump onto the slightly higher
surface and walk to (5). Fortunately, they cannot -- but this point
will be important later.

$ $

\emph{Add 1 Gadget: }A cannon (on the left side of the image) constantly
shoots out monstars. They all bounce off the platform at (2) and die
on the spikes below.

The gadget activates when Tim pulls a lever that controls the platform
at (2). When this happens, the platform at (2) rises for just long
enough to let one monstar go under it, and then falls back into place.

Every time this happens, the one monstar that passes the gadget walks
through the one-way wall at (3) and gets trapped at (1). Therefore,
whenever Tim activates the Add 1 Gadget, the value of the counter
is increased by 1.

$ $

\begin{center}
\includegraphics[scale=0.5]{counter_gadget}
\par\end{center}

\begin{center}
{\footnotesize{(The same picture again, so you don't have to keep
scrolling between pages)}}
\par\end{center}{\footnotesize \par}

\emph{Remove 1 Gadget: }A cannon (at the bottom of the image) constantly
shoots out bunnies. Similarly to the Add 1 Gadget, they all bounce
off the platform and die on the spikes below.

The gadget activates when Tim pulls a lever that controls its platform.
When this happens, the platform rises for just long enough to let
one bunny pass, and then falls back into place. When this happens,
the bunny shoots up to (1).

In Braid, monstars bounce off of bunnies, killing them. When this
gadget activates, if there are any monsters at (1), the bunny shoots
into the monstars' feet. This kills the bunny, and bounces exactly
one monstar onto the slightly higher platform. This one monstar then
leaves the counter, walking away to (5).

If there are no monstars in the counter, then the bunny will shoot
up into the spikes at (4) and die.

The overall effect of this gadget is: when it is activated, if the
value of its counter is not zero, then one monstar is removed from
the counter and sent down a separate path.

$ $

This completes the description of the Counter Gadget.

Note: in the Remove 1 Gadget, it is important that the removed monstar
is kept alive and sent down a separate path. We will use it later
with the Branch Gadget, below.

\subsubsection*{Branch Gadget}

The goal of this gadget is to force Tim down one of two different
paths, depending on whether there is a monstar in the gadget.

\begin{center}
\includegraphics[scale=0.5]{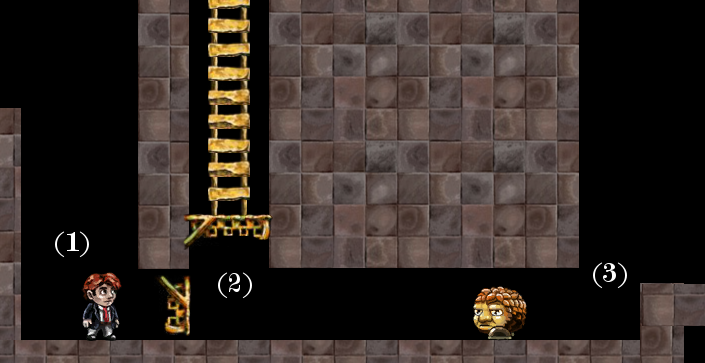}
\par\end{center}

Tim enters the gadget at (1). The platform and ladder above (2) are
too high for him to jump to. So, if there is no monstar in the gadget,
then his only option is to exit at (3).

If there is a monstar in the gadget (as is the case in the picture),
then Tim cannot walk to (3), because the monstar is blocking his way.
The only way for Tim to proceed is to jump on the monstar, killing
it. The only place where he has enough vertical space to do this is
at (2). When he jumps on the monstar at (2), he bounces high enough
to land on the one-way platform above (2). After this, his only option
is to exit the gadget by climbing the ladder.

Therefore, depending on whether there is a monster in the gadget or
not, Tim is forced to exit either up the ladder, or at (3). This completes
the Branch Gadget.

\subsection{Assembling the Counter Machine}

Now that we have our gadgets, we will describe how to use them to
simulate a counter machine with three counters.

First of all, place three Counter Gadgets, far away from each other.
These will be our counters.

Now, construct a large course for Tim to walk through, consisting
of Lever Pull Gadgets and Branch Gadgets. The structure of the course
determines the program which the counter machine will run. For example,
if we want our program to start with ``Add 1 to counter \#2'', then
the first room of the course will be a Lever Pull Gadget whose lever
connects to the Add 1 Gadget inside the 2$^{\text{nd}}$ Counter Gadget.

This handles ``Add'' operations. What about ``Subtract and Branch
if 0'' operations? The ``Subtract'' part is easy: force Tim through
a Lever Pull Gadget whose lever connects to a counter's Remove 1 Gadget.

Suppose Tim activates this gadget. Further suppose the counter is
not 0. Then upon activating the gadget, one monstar will exit the
counter. To perform the ``Branch if 0'' part of the operation, we
would like to guide this monstar into the appropriate Branch gadget.

To do this, build a room that catches all monstars leaving the counter.
This room should consist of a flat floor for the monstar to walk back
and forth across, along with many platforms forming trap doors in
the floor.

Next, force Tim through a second Lever Pull Gadget, whose lever opens
one of these trap doors. Build a path below this trap door, leading
the monstar to the appropriate Branch gadget.

Finally, direct Tim to this Branch gadget. Make sure to make Tim's
path long enough that he cannot beat the monstar to the gadget. Depending
on whether the counter was nonzero or zero, he will or will not find
a monstar in the Branch gadget, and he will be forced down one of
two different paths.

This completes the ``Subtract and Branch if 0'' operation.

We have handled ``Add'' operations and ``Subtract and Branch''
operations. All that remains is the ``Halt'' operation. This is
the simplest operation to simulate: just build a room that contains
the end goal, e.g. a puzzle piece and an exit door.

$ $

We have shown how to completely simulate a counter machine inside
Braid. As noted earlier, such a counter machine is Turing-complete.
Therefore, Braid\emph{ }is Turing-complete. This proves the desired
result.
\begin{thm}
Deciding whether a given Braid level is solvable is as hard as the
halting problem. $\square$
\end{thm}

\section{``Bounded Braid'' is decidable}

In the introduction, we conjectured another way that Braid could be
undecidable. Namely, one could use rewind data to store large amounts
of information. Concievably, one might use Braid's rewind data as
the tape of a very large Turing machine.

Such a construction would be more satisfying that the counter machine
simulation, because it relies on Braid's time manipulation, the puzzle
mechanic that makes Braid unique.

Surprisingly, it turns out that such a strategy cannot work. In Theorem
2, we prove that if we bound Braid's gameplay to eliminate counter
machine strategies, then Braid becomes decidable.
\begin{thm}
\emph{Suppose we bound Braid's gameplay so that at most $2^{n}$ game
elements can exist at once, and all the action takes place inside
an $n$ pixel by $n$ pixel box. Then the problem of determining whether
a level is beatable is in 2-$\mathsf{EXPSPACE}$. That is, it is solvable
in $2^{2^{poly(n)}}$ space.}\end{thm}
\begin{proof}
Note that, with these bounds on Braid's gameplay, there are only $2^{poly(n)}$
game states for any given timestep.

We will model the entire game as a special kind of Turing machine.

In this Turing machine, the tape will record the rewind data. To be
specific, the symbol at cell $i$ in the tape contains the data for
all the time-dependent (i.e. living inside of time) objects that exist
on timestep $i$.

The head of the Turing machine will track the current time. If the
player is currently viewing timestep $i$, whether by normal gameplay
or by rewinding, then the head will be at cell $i$. The state of
the head contains the data for all the time-immune objects that currently
exist.

This is a non-deterministic Turing machine. The player may control
both their normal game movements (like running and jumping), and the
movement of the head (by rewinding time or playing it forward)%
\footnote{In Braid, the player can rewind time or play it forward at multiple
different speeds, ranging from 0 (time is paused) to 8 (rapid rewinding).
Since 8 is a constant, it is easy to modify the Turing machine construction
appropriately.%
}. The player's possible choices, combined with the game's physics
engine, produce the nondeterministic state transition table for the
Turing machine.

This Turing machine has one very special property. \emph{When the
head writes a symbol, all data to the right of that symbol is erased.}
Why does this happen? Because, as stated in Section 2.2, if Tim ``undoes''
something and then makes another action, the option to ``redo''
is lost forever.

Let us make this more precise. Our Turing machine writes a symbol
exactly when it is recording rewind data, i.e, when the player is
not traveling through time by holding the Shift key. As stated in
Section 2.2, when the player leaves time traveling mode, the rewind
data for Tim's future is deleted. Braid's rewind data corresponds
exactly to the symbols on the machine's tape. Therefore, whenever
our Turing machine writes a symbol, all data for the future, i.e.
all data to the right of the head, is deleted.

This Turing machine perfectly simulates Braid. As one final addition,
we will add a target state, $T$, to our machine. The machine enters
state $T$ when Tim achieves the goal of the level, i.e. collecting
all the puzzle pieces and entering the exit door.

The Braid level is solvable if and only if our Turing machine can
reach state $T$. Therefore, we have reduced Braid to the problem
of deciding whether a special type of Turing machine can reach a specified
target state.

For this theorem, we assume that only $2^{n}$ objects can exist at
once, and that the action must take place within an $n$-by-$n$ grid.
With these bounds, there are only $2^{poly(n)}$ possible game states
for any given timestep. Therefore, our Turing machine has $2^{poly(n)}$
states and $2^{poly(n)}$ symbols.

We will defer the meat of the proof, a technical result about these
data-erasing Turing machines, to Theorem 4 in the next section. The
desired result, that Braid is in 2-$\mathsf{EXPSPACE}$, follows.
\end{proof}

\section{Braidlike Turing machines are decidable}
\begin{defn}
A \emph{braidlike} Turing machine is a Turing machine with the following
special property: whenever the head writes a symbol, all data to the
right of that symbol is erased.\end{defn}
\begin{thm}
Consider the following decision problem: we are given a nondeterministic
braidlike Turing machine with $N$ states and $S$ symbols. We want
to determine whether the machine can reach a specified ``target state''
$T$. This decision problem is solvable in $2^{poly(N)}\log^{2}(S)$
space.
\end{thm}
In order to motivate the proof of Theorem 4, we will first prove two
easier theorems.

For the theorems in this section, we assume all Turing machines to
be half-infinite. In other words, the tape has a left endpoint, and
extends infinitely far to the right. We number the cells 0, 1, 2,
..., et cetera. This is only for our personal convenience; it is not
hard to modify all of these proofs to work for a doubly infinite tape.

Here is our first ``warm-up'' lemma:
\begin{lem}
A read-only Turing machine has decidable behavior.\end{lem}
\begin{proof}
WLOG, the Turing machine head starts at cell 0. We assume that the
tape contains a finite string starting at cell 0 and ending at cell
$n$, and that the rest of the tape is blank.

Whenever the machine head moves from cell $i$ to cell $i+1$, we
will employ a ``tour guide'' to stand between the two cells. The
tour guide's job is to save the machine head the hassle of moving
left. Whenever the head attempts to move from cell $i+1$ back to
cell $i$, the tour guide will interrupt. She will say, ``Eventually,
you will go back to cell $i+1$, and the first time you do, you will
be in state $X$.'' The machine head then moves back to cell $i+1$
in state $X$, having skipped many tiresome steps of calculation.

(Note: a tour guide may also (1) immediately halt the Turing machine
and either ACCEPT or REJECT, or (2) inform the Turing machine that
it will loop forever.)

It is possible to compute the behavior of any given tour guide. In
fact, in order to compute the behavior of the tour guide between cells
$i$ and $i+1$, we need only know (1) the behavior of the tour guide
immedately to her left, and (2) the symbol on cell $i$.

These tour guides ensure that the machine head never moves to the
left. Such a machine is obviously decidable.\end{proof}
\begin{rem*}
An easy modification of this proof shows that read-only Turing machines
can only parse regular languages.

Now we prove a slightly more involved lemma.\end{rem*}
\begin{lem}
A deterministic braidlike Turing machine is decidable. In particular,
if the machine has $N$ states and $S$ symbols, and starts on a blank
tape, then we can determine its behavior in $2^{poly(N)}\log S$ space.\end{lem}
\begin{proof}
As in the previous lemma, we use tour guides. Again, a tour guide
is created whenever the machine head moves from cell $i$ to cell
$i+1$. The tour guides behave as in the previous lemma, with four
differences.

(a) In this lemma, the Turing machine starts on a blank tape.

(b) When the machine head writes a symbol, not only is all data to
its right erased, but all the tour guides to its right are destroyed.

(c) When the head moves from cell $i+1$ to cell $i$, the tour guide
between the cells still interrupts. She now has one new possible response:
``Eventually, before visiting cell $i+1$ again, you will destroy
me.'' After this response, the head moves to cell $i$, and continues
calculating as if the tour guide didn't exist.

(d) Each tour guide also remembers one more piece of information:
the state that the machine head was in when she was first created.

Suppose that at some point, two identical tour guides exist. Call
the one further to the left Alice, and the one further to the right
Bob. By ``identical'', I mean that Alice and Bob would give identical
responses when the machine head moves past them in any state $X$,
and also that they remember the same information with regards to point
(d).

Consider the time period starting at Alice's creation. By assumption,
the machine started on a blank tape. Therefore, when Alice was created,
the entire tape to the right of Alice was blank. Also, since Alice
has not yet destroyed, the machine head never went to the left of
Alice. (That is, whenever it tried to go to the left of Alice, Alice
always interrupted it and sent it back.)

Now, we can draw out a timeline of the machine's operation:

$\bullet$ It performed a series of steps $\mathcal{X}_{1}$, ending
with Alice's creation.

$\bullet$ It performed a series of steps $\mathcal{X}_{2}$, never
going to the left of Alice, ending in Bob's creation.

But Alice and Bob are identical! Therefore, to the machine head, the
situation after $\mathcal{X}_{1}$ looks identical to the situation
after $\mathcal{X}_{2}$: there is a tour guide identical to Alice
immediately to its left, and a blank tape stretching out to infinity
to its right. The machine is even in the same state, by point (d).

Therefore, after creating Bob, the machine will perform the same series
of steps $\mathcal{X}_{2}$ again. And after that, it will create
a third identical tour guide, and repeat $\mathcal{X}_{2}$ yet again.
This process will repeat forever. Therefore, the machine will loop
forever, never entering any states it has not entered before.

Now, we will use the Pigeonhole Principle. How many possible tour
guides are there? Each tour guide is essentially a map from the $N$
states of the machine to all $N+4$ possible responses (one response
for each state, in addition to ``ACCEPT'', ``REJECT'', ``loop
forever'', and ``destroy me''). Each tour guide also remembers
one of $N$ possible initial states (see point (d), above). Thus,
there are $(N+4)^{N}N=2^{poly(N)}$ different tour guides.

Therefore, if the Turing machine ever reaches cell $(N+4)^{N}N+1$,
there is no need to continue: we know that two identical tour guides
must exist, so the machine must loop forever. So, we can efficiently
simulate this Turing machine using only $2^{poly(N)}$ cells. This
requires $2^{poly(N)}\log(S)$ space. By also storing a ``total number
of steps'' counter of size $2^{poly(N)}\log(S)$ bits, we can detect
infinite loops within this space. Therefore, we can decide in space
$2^{poly(N)}\log(S)$ the complete behavior of this Turing machine.
\end{proof}
We're almost done. All that remains is to add nondeterminism to the
previous lemma. We now restate and finally prove Theorem 4.
\begin{thm*}
A nondeterministic braidlike Turing machine has decidable behavior.
In particular, if the machine has $N$ states and $S$ symbols, and
starts on a blank tape, then we can decide whether it can reach a
given target state $T$ in $2^{poly(N)}\log^{2}(S)$ space.\end{thm*}
\begin{proof}
Again, we use tour guides. There are two changes to the tour guides'
behavior.

(a) Instead of giving a single response (like ``you will return to
cell $i+1$ in state $X$''), the tour guide now gives a set of possible
responses (like ``you can return to cell $i+1$ in state $X_{1}$,
$X_{2}$, ..., or $X_{n}$; or you can destroy me, or you can loop
forever.'') This captures the machine's nondeterminism.

(b) When a tour guide is created, the machine head is required to
inform it of its destiny. The machine head must tell it either (1)
``You will survive forever'', or (2) ``I will destroy you, and
the last time I am in the cell immediately to your right, I will be
in state $X$.''

Now we show that it is never necessary for two identical tour guides
to exist. Again, by ``identical'', I mean that the tour guides would
give identical responses to a machine head in any state $X$, and
also that they remember the same information with regards to point
(d) in the previous lemma and point (b) in this lemma.

Suppose that at some point during the machine's operation, two identical
tour guides exist, and later the machine reaches the target state
$T$. Then I claim that the machine could have reached $T$ in strictly
fewer steps.

As in the previous lemma, pick a time when two identical tour guides
exist. Name the one further to the left Alice, and the one further
to the right Bob.

\uline{Case 1}: Alice and Bob both survive forever. Then, as in
the previous lemma, the situation after creating Alice is identical
to the situation after creating Bob. Therefore, the steps between
creating Alice and creating Bob were unnecessary.

\uline{Case 2}: Alice and Bob will both be destroyed, and the last
times the machine head is in the cells immediately to their right,
it will be in state $X$. Then we can map out the machine's operation,
like so:

$\bullet$ Perform a series of steps $\mathcal{X}_{1}$, ending with
Alice's creation.

$\bullet$ Perform a series of steps $\mathcal{X}_{2}$, ending with
Bob's creation.

$\bullet$ Perform a series of steps $\mathcal{X}_{3}$, ending immediately
to the right of Bob in state $X$ for the last time.

$\bullet$ Perform a series of steps $\mathcal{X}_{4}$, ending immediately
to the right of Alice in state $X$ for the last time.

$\bullet$ Perform a series of steps $\mathcal{X}_{5}$, ending in
the target state $T$.

Is this optimal? No! I claim that the machine should have applied
the $\mathcal{X}_{3}$ strategy immediately after creating Alice.
Why is $\mathcal{X}_{3}$ a valid series of steps immediately after
creating Alice? Because, while performing $\mathcal{X}_{3}$ in the
above timeline, the machine never goes to the left of Bob. (That is,
whenever it tried to go to the left of Bob, Bob sent it back to the
right.) Therefore, we can use the exact same logic as in the previous
lemma. That is, to a machine head who wants to perform $\mathcal{X}_{3}$,
the situation immediately after creating Alice looks the same as the
situation immediately after creating Bob. 

Why is this important? Because the machine could have performed $\mathcal{X}_{1}$,
then $\mathcal{X}_{3}$, then $\mathcal{X}_{5}$, and reached state
$T$ in strictly fewer steps.

In both Case 1 and Case 2, having two identical tour guides is redundant.
This proves the claim.

Therefore, we can assume that the machine never has two identical
tour guides. Again, we proceed by Pigeonhole.

How many tour guides are there? Each tour guide is essentially a map
from the $N$ states to all $2^{poly(N)}$ possible sets of responses.
There are $\left(2^{poly(N)}\right)^{N}=2^{poly(N)}$ such maps. Each
tour guide also knows about her creation and her destiny, but this
only contributes a $poly(N)$ factor to the number of tour guides.
Therefore, there are $2^{poly(N)}$ possible tour guides.

Therefore, it suffices to use $2^{poly(N)}$ cells, i.e. $2^{poly(N)}\log(S)$
space, to simulate this nondeterministic braidlike Turing machine.
Therefore, the decision problem ``Can the machine reach the target
state $T$?'' is contained in $\mathsf{NSPACE}\left(2^{poly(N)}\log(S)\right)$.
By Savitch's Theorem, this is contained in $\mathsf{SPACE}\left(2^{poly(N)}\log^{2}(S)\right)$.
\end{proof}

\section{Open Questions}

Bounded Braid is somewhere between $\mathsf{PSPACE}$ and 2-$\mathsf{EXPSPACE}$.
Where? I conjecture that it hits the 2-$\mathsf{EXPSPACE}$ ceiling.
In particular:
\begin{conjecture*}
Bounded Braid is complete for the problem of simulating an (exponentially
larger) braidlike Turing machine.
\end{conjecture*}
And more interestingly,
\begin{conjecture*}
The halting problem for braidlike Turing machines is $\mathsf{EXPSPACE}$-complete.
\end{conjecture*}
$ $

There is also a more ``open'' open question:
\begin{problem*}
What other video games are undecidable? Our construction in Section
4 should apply to many video games in which one can spawn an unbounded
number of objects, and where it is possible to remove exactly one
of these objects from a crowd.
\end{problem*}

\bibliographystyle{plain}
\nocite{*}
\bibliography{Braid_is_undecidable}

\begin{thebibliography}{1}

\bibitem{DBLP:journals/corr/abs-1203-1895}
Greg Aloupis, Erik~D. Demaine, and Alan Guo.
\newblock Classic {N}intendo {G}ames are ({NP}-){H}ard.
\newblock {\em CoRR}, abs/1203.1895, 2012.

\bibitem{story}
Luke Arnott.
\newblock Unraveling {B}raid: {P}uzzle {G}ames and {S}torytelling in the
  {I}mperative {M}ood.
\newblock {\em Bulletin of Science Technology Society}, 32:433--440, 2012.

\bibitem{Pspace-complete97sokobanis}
Joseph~C. Culberson.
\newblock Sokoban is {PSPACE}-complete, 1997.

\bibitem{Flake_rushhour}
Gary Flake and Eric Baum.
\newblock Rush {H}our is {PSPACE}-complete, or "{W}hy you should generously tip
  parking lot attendants", 2002.

\bibitem{counters}
Marvin Minsky.
\newblock Recursive {U}nsolvability of {P}ost's {P}roblem of '{T}ag'.
\newblock {\em Annals Of Mathematics}, 74:437--453, 1961.

\bibitem{DBLP:journals/corr/abs-1201-4995}
Giovanni Viglietta.
\newblock Gaming is a hard job, but someone has to do it!
\newblock {\em CoRR}, abs/1201.4995, 2012.

\bibitem{DBLP:journals/corr/abs-1202-6581}
Giovanni Viglietta.
\newblock Lemmings is {PSPACE}-complete.
\newblock {\em CoRR}, abs/1202.6581, 2012.

\end{thebibliography}

\end{document}